\documentclass[letterpaper,10pt,conference]{ieeeconf}
\IEEEoverridecommandlockouts
\usepackage{cite}
\usepackage{hyperref} % Clickable links.
\hypersetup{colorlinks=true} % Links are colored by type.
\usepackage{amsfonts,latexsym,graphicx,comment,url,amssymb,dsfont,amsmath,booktabs,xcolor,cuted, siunitx}
\usepackage[letter,center]{crop}
\usepackage[english]{babel}
\usepackage[utf8]{inputenc}
\usepackage[norefs,nocites]{refcheck}
\usepackage[noframe]{showframe}
\usepackage[keeplastbox]{flushend}
\usepackage{subcaption}
\usepackage{multirow}
\usepackage{floatrow}
\usepackage{comment}
\usepackage{mathrsfs}

\renewcommand{\Re}{\mathbb{R}}
\renewcommand{\paragraph}[1]{\smallskip\noindent\textbf{#1.} }
\newtheorem{prop}{\bf Proposition}
\newtheorem{lem}{\bf Lemma}
\DeclareMathOperator*{\argmin}{arg\,min}

\newfloatcommand{capbtabbox}{table}[][\FBwidth]
\usepackage{bibentry}
\nobibliography*

\newtheorem{definition}{\bf Definition}

\newtheorem{remark}{Remark}

\definecolor{joblue}{rgb}{0,0,0.8}
\definecolor{jored}{rgb}{0.8,0,0}
\definecolor{jogreen}{rgb}{0.1,.8,0.1}

\definecolor{Mgreen}{rgb}{0,1,0}

\title{Learning Reduced Nonlinear State-Space Models: an Output-Error Based Canonical Approach}
\author{Steeven Janny$^{1}$, Quentin Possamaï$^{1}$, Laurent Bako$^{3}$, Madiha Nadri$^{2}$ and Christian Wolf$^{4}$ % <-this % stops a space
\thanks{$^{1}$Quentin Possamaï and Steeven Janny are with INSA-Lyon, LIRIS UMR CNRS 5205.
{\tt\small quentin.possamai@insa-lyon.fr},
{\tt\small steeven.janny@insa-lyon.fr}}%
\thanks{$^{2}$Madiha Nadri is with Université Claude Bernard Lyon 1, LAGEPP
{\tt\small madiha.nadri-wolf@univ-lyon1.fr}}%
\thanks{$^{3}$Laurent Bako is with Univ Lyon, Ecole Centrale de Lyon, INSA Lyon, Université Claude Bernard Lyon 1, CNRS, Ampère,
UMR5005, 69130 Ecully, France,
{\tt\small laurent.bako@ec-lyon.fr}}%
\thanks{$^{4}$Christian Wolf is with Naver Labs Europe, Meylan, France,
{\tt\small christian.wolf@naverlabs.com}}%
}

\begin{document}

\maketitle
\thispagestyle{empty}
\pagestyle{empty}

\begin{abstract}
The identification of a nonlinear dynamic model is an open topic in control theory, especially from sparse input-output measurements. A fundamental challenge of this problem is that very few to zero prior knowledge is available on both the state and the nonlinear system model.  To cope with this challenge, we investigate the effectiveness of deep learning in the modeling of dynamic systems with nonlinear behavior by advocating an approach which relies on three main ingredients: (i) we show that under some structural conditions on the to-be-identified model, the state can be expressed in function of a sequence of the past inputs and outputs; (ii) this relation which we call the state map can be modelled by resorting to the well-documented approximation power of deep neural networks; (iii) taking then advantage of existing learning schemes, a state-space model can be finally identified. After the formulation and analysis of the approach, we show its ability to identify three different nonlinear systems. The performances are evaluated in terms of open-loop prediction on test data generated in simulation as well as a real world data-set of unmanned aerial vehicle flight measurements.
\end{abstract} 

\textbf{Keywords:} nonlinear system identification, state-space models, model reduction, deep learning,  auto-encoding

\section{Introduction}

\noindent
A large majority of deployed methods from control theory have as a prerequisite the provision of relatively precise dynamic model characterizing the temporal evolution of the state variables at stake. This model generally plays a central role, since the performance of the method is often directly related to the accuracy of the model \cite{weinmann2012uncertain,cheah2006adaptive,bauersfeld_neurobem_2021, busoniu_reinforcement_2018, bemporad2006model}. Consequently, modeling and identification of a dynamic system is an essential preliminary step, since it will serve as foundation for additional processing, such as controller or observer design. However, this is not a trivial task: in the general case, the system is complex, non-linear, and involves physical phenomena that are often difficult or even impossible to model correctly without strongly impacting the required computation time. On the other hand, the identification of the parameters of a non-linear model is a non-convex problem, which can require tremendous hours of calibrations and experiments. Moreover, identification of a dynamic system often requires the intervention of domain experts and the ability to freely interact with the system.

%However, most classical approaches identify nonlinear models in input/output form, without giving an explicit representation of the state space often needed for advanced controller synthesis as well as observers \mn{ref des methodes classiques }.

The development of data-driven techniques for the identification of non-linear systems has provided a promising response to these issues and has received great interest over the last decade (see for example \cite{LJUNG20201175}, \cite{MASTI2021}, \cite{PILLONETTO2014657}). Specifically, neural networks propose to remove the burden of modeling by replacing it with the collection of massive datasets from the system of interest. Modeling methods based on Deep Learning constitute an alternative solution to painstaking physical modeling. The main idea is based on the use of an extremely versatile model, capable of approaching most dynamics with a certain degree of precision, that can be directly identified from pairs of input-output measurements in the case of dynamic systems, provided that theses measurements gather enough information necessary to approximate the true dynamic. Nevertheless, the great flexibility of neural networks comes at the cost of a lack of mathematical structure making it difficult to perform theoretical analysis in terms of robustness, precision and stability. Moreover, learning complex, high-dimensional dynamical systems is not straightforward. The general formulation leads to latent dynamic models lacking of meaningful physical structure and requires large dimensional state spaces.

In this article, we propose a new identification structure for nonlinear state-space systems from a set of observation trajectories and associated inputs. We demonstrate the existence of a regressor inspired by finite impulse response models allowing to map a series of past observations to future outputs, and provide bounds derived from the prediction error during deployment. We then deduced a high-dimensional canonical state-space model discovered using an output-error based approach and propose to learn an auto-encoder projecting the dynamics into a smaller state-space. We evaluate our proposal on different systems in simulation and in the real world.

\section{Related Work}

\noindent
Data-driven dynamic models are widely studied in the community and get a lot of attention. In particular, \cite{brunton2016discovering,Sahoo2018Learning,Chen2021physicsinformed} propose to find governing physical equations by performing a sparse regression from the data. At the junction between physical model and learning, \cite{yin2021augmenting,Long2018HybridNetIM,qi2019integrating,mehta2021neural} use neural networks to model complementary phenomena not described by the initial physical model. In particular, \cite{shi_neural_2019,bauersfeld_neurobem_2021,possamai2022learning} extend the dynamic model of a unmanned aerial vehicle with a neural network in charge of predicting aerodynamic disturbances, which are often very demanding and intractable for real time physical simulation.

Close to our work, a body of literature proposes to use deep learning for the identification of latent models, that is to say without direct physical meaning. This is notably the case for recent works around the Koopman operator \cite{lusch2018deep,janny2021deepkkl,peralez2020datadriven,rowley2009spectral,buissonfenet2022towards}. Another solution is to use an auto-encoder structure to model the latent dynamics of a system from past observations, following \cite{MASTI2021,beintema2021}. Our proposal differs from this line of works in three main points: (1) we provide theoretical results and conditions for the existence of the dynamical system that we identify, (2) we propose to use a high-dimensional regressor structure without explicit state representation, which will be deduced from a dimensionality reduction operation and (3) we evaluate our approach on challenging and unstable systems.

% \noindent
% \textbf{Dimension reduction and auto-encoders} -- 

\paragraph{Notation} $\Re$ and $\mathbb{N}$ denote the sets of reals and natural numbers respectively. $\left\|\cdot\right\|$ refers to a generic norm on some appropriate space. 

\section{Problem statement and preliminary results}
\subsection{Problem statement}

\noindent
We consider a nonlinear discrete-time system of the general form
\begin{equation} \label{eq_general_form}
    \left\{
\begin{array}{rcl}
     x_{t+1} &=& f^\circ(x_t, u_t) \\
     y_t &=&  h^\circ(x_t,u_t) +w_t,
\end{array}
\right.
\end{equation}
with $x_t\in \mathcal{X}\subset\Re^{n_x}$, $u_t\in \mathcal{U}\subset \Re^{n_u}$, $y_t\in \mathcal{Y}\subset \Re^{n_y}$ being the state, the input and the output of the system at discrete time $t\in \mathbb{N}$ respectively. $f^\circ:\Re^{n_x}\times \Re^{n_u}\rightarrow \Re^{n_x}$ and $h^\circ:\Re^{n_x}\times \Re^{n_u}\rightarrow \Re^{n_y}$ are some nonlinear vector-valued functions. As to $w_t\in \mathcal{W}\subset \Re^{n_y}$, it represents measurement noise.   
We will make the following  important assumptions: 
\begin{enumerate}
\item The external signals $u$ and $w$ take values in compact sets $\mathcal{U}$ and $\mathcal{W}$ respectively.
\label{assump:Compacts}
\item The state-space $\mathcal{X}$ is a known compact set  containing the initial state $x_0$. 
\item $(\mathcal{X},\mathcal{U},\mathcal{W},\mathcal{Y})$ and  $(f^\circ,h^\circ)$ satisfy the following invariance conditions:
$$
\begin{aligned}
& \forall (x,u)\in \mathcal{X}\times \mathcal{U}, f^\circ(x,u)\in \mathcal{X}  \\
& \forall (x,u,w)\in \mathcal{X}\times \mathcal{U}\times \mathcal{W}, h^\circ(x,u)+w \in \mathcal{Y} 
\end{aligned}
$$
\label{assump:invariance}
%initial state $x_0\in \mathcal{X}$ and for any sequences $\left\{u_t\right\}\subset \mathcal{U}$ and  $\left\{w_t\right\}\subset \mathcal{W}$ the state and output trajectories remain bounded in compact sets $\mathcal{X}$ and $\mathcal{Y}$
%\item 

\item $f^\circ$ (and $h^\circ$) are uniformly Lipschitz continuous on $\mathcal{X}\times \mathcal{U}\subset \Re^{n_x}\times \Re^{n_u}$ with respect to $\mathcal{U}$, i.e., there exists a constant $\gamma_f>0$ such that  $\|f^\circ(x,u)-f^\circ(x',u)\|\leq \gamma_f\|x-x'\|$ for all $(x,x',u)\in \mathcal{X}\times \mathcal{X}\times \mathcal{U}$.  
\label{assump:Lipschitz}
\end{enumerate}	
%where the state $x \in \mathbb{R}^{d_x}$ $u \in \mathbb{R}^{d_u}$, the control input,  $y \in \mathbb{R}^{d_y}$ the output, $F$ and $h$ are suitable and unknown functions.\\ Although we assume that the collected input and output signals may be affected by measurement noise.
The assumptions \ref{assump:Compacts}--\ref{assump:Lipschitz} are required essentially to theoretically ensure the well-definedness of optimization problems that will be expressed later in the paper. Of course in the context of system identification,  such types of assumptions are not intended to be checked prior to applying the method to be developed.

The problem of interest in this paper can be stated as follows: 
Given a finite number $N$ of input-output data pairs $\left\{(u_t,y_t):t=1,\ldots,N\right\}$ generated by a nonlinear system of the form \eqref{eq_general_form}, find an appropriate dimension $n_x$ of a state-space representation along with estimates of the associated  functions $f^\circ$ and $h^\circ$. 
Here, the number $n_y$ of outputs and the number $n_u$ of inputs  are known a priori. However the dimension $n_x$ of the state is a parameter of the model which needs to be estimated along with the maps $(f^\circ,h^\circ)$. 

We develop a solution in three steps: first, a nonlinear regression model is derived from the data-generating system equations in   \eqref{eq_general_form}. The underlying nonlinear map is then modelled by a deep neural network structure and trained with the available data following an output-error framework. Given this map, one can readily form an equivalent canonical state-space representation of \eqref{eq_general_form} with, however, the drawback that its dimension may be high. Hence, the third and last step of the proposed procedure consists in model reduction\footnote{By  state-space model reduction, we mean the reduction of the state dimension. The process aims at finding  another state-space model which is as close as possible to the primary one but with a compressed state. }, an objective which is achieved through the design of an appropriate  encoder-decoder.

\subsection{Preliminary results}
\label{subsec:preliminary_results}
\noindent An important challenge concerning the identification of the system \eqref{eq_general_form} is the fact that the state $x_t$ is not entirely measured. We therefore need to express it first as a function of the available past input-output measurements $\left\{(u_\tau,y_\tau): \tau<t\right\}$. 
Indeed, if the noise $w_t$ in  \eqref{eq_general_form} is assumed to be identically equal to zero, then under appropriate observability conditions on the system, there exists a time horizon $\ell$ and a map $\phi:\Re^{L}\rightarrow \Re^{n_x}$, with $L=\ell(n_u+n_y)$, such that the state $x_t$ can be written as
\begin{equation}\label{eq:state-function}
x_t = \phi(z_t) 
\end{equation}
where 
\begin{equation}\label{eq:zt}
z_t=\begin{pmatrix}u_{t-\ell}^{\top} & y_{t-\ell}^{\top}&  \cdots  & u_{t-1}^{\top} & y_{t-1}^{\top}\end{pmatrix}^\top
%z_t=\begin{pmatrix}y_{t-1}^{\top} & \cdots & y_{t-\ell}^{\top} & u_{t-1}^{\top} & \cdots & u_{t-\ell}^{\top}\end{pmatrix}^\top
\end{equation}
is the so-called regressor vector. 
To show the existence of such a map $ \phi $, some observability conditions on the system to be identified \eqref{eq_general_form} are needed. For this purpose let us start by introducing some notations. 
For a positive integer $i$, let $F_i:\Re^{n_x}\times \Re^{in_u}\rightarrow \Re^{n_x}$ be the map defined recursively from the function $f^\circ$ in \eqref{eq_general_form} as follows: for $x\in \Re^{n_x}$ and $(u_1,\ldots,u_i)\in \Re^{in_u}$, 
$F_1(x,u_1)=f^\circ(x,u_1)$ and for all $i\geq 2$, 
\begin{equation}\label{eq:Fi}
F_i(x,u_1,\ldots,u_i)=f^\circ\big(F_{i-1}(x,u_1,\ldots,u_{i-1}),u_i\big).
\end{equation}
%%%%%
Before proceeding further, let us mention a useful property of the maps $F_i$. 
\begin{lem}\label{lem:Lipschitz}
Under Assumption \ref{assump:invariance}, if $f^\circ:\Re^{n_x} \times \Re^{n_u}\rightarrow \Re^{n_x}$ is uniformly $\gamma_f$-Lipschitz on $\mathcal{X}\times \mathcal{U}$ with respect to $\mathcal{U}$, then the map $F_i$ defined in  \eqref{eq:Fi} is uniformly $\gamma_f^i-$Lipschitz on $\mathcal{X}\times\mathcal{U}^i$ with respect to  $\mathcal{U}^i\subset \Re^{in_u}$. 
\end{lem}
%%%%
\begin{proof}
The proof of this lemma is straightforward and is therefore omitted. 
\end{proof}
%%%%
Now consider the function $\mathcal{O}_i:\Re^{n_x}\times \Re^{in_u}\rightarrow \Re^{in_y}$ given by 
\begin{equation}\label{eq:Observability-Function}
\mathcal{O}_i(x,u_1,\ldots,u_i) =  \begin{pmatrix}h^\circ(x,u_1)\\ h^\circ\big(F_{1}(x,u_1),u_2\big)\\ \vdots \\  h^\circ\big(F_{i-1}(x,u_1,\ldots,u_{i-1}),u_i\big)\end{pmatrix}.
\end{equation}
For notational simplicity, let us pose $\bar{u}_{1|i}=(u_1,\ldots,u_i)$ so that $\mathcal{O}_i(x,u_1,\ldots,u_i)$ in  the previous equality can be replaced by $\mathcal{O}_i(x,\bar{u}_{1|i})$. 
%To investigate the existence of the state map $\phi$ in \eqref{eq:state-function}, we need to require an observability property for the system \eqref{eq_general_form} (or, equivalently, an injectivity property of the maps $\mathcal{O}_i$). 
%%%%
\begin{definition}\label{def:observability}
The system \eqref{eq_general_form} is said to be  finite-time observable over a time horizon $r\in \mathbb{N}$  if for each $\bar{u}\in \mathcal{U}^{r }$, the function $\mathcal{O}_r(\cdot,\bar{u})$, with $\mathcal{O}_r$ defined as in \eqref{eq:Observability-Function}, is injective. 
%In the case where this property holds for $\mathcal{U}=\Re^{n_x}$ the system is said to be \textit{globally} finite-time observable. 
\end{definition}

Note that if the observability property in Definition \ref{def:observability} holds for some $r\in \mathbb{N}$ then it holds as well for any $i\geq r$. 

\begin{prop}[Existence of the map $\phi$]
If the nonlinear system \eqref{eq_general_form} (considered under the assumption that $w\equiv 0$) is  finite-time observable in the sense of Definition \ref{def:observability}, then there exist $\ell\in \mathbb{N}$ and a (nonlinear) map $\phi:\Re^{L}\rightarrow \Re^{n_x}$ such that 
\eqref{eq:state-function} holds 
for all time $t\geq \ell$, any initial state  in $\mathcal{X}$ and any input signal taking values in $\mathcal{U}$.  
\end{prop}
%%%%

\begin{proof}
For discrete time indices $(i,j)$ with $i\leq j$, let $\bar{y}_{i|j}=\begin{pmatrix}y_i^\top & \cdots & y_j^\top\end{pmatrix}^\top$ be a vector of outputs from  time $i$ to time $j$. Likewise define $\bar{u}_{i|j}=\begin{pmatrix}u_i^\top & \cdots & u_j^\top\end{pmatrix}^\top$.

By iterating the system equations, it is easy to see that 
%$$\bar{y}_{t-\ell|t-1} = h(F_{\ell-1}(x_{t-\ell},\bar{u}_{t-\ell|t-2}),u_{t-1}) $$
%which reads as
\begin{equation}\label{eq:DataEq}
	\bar{y}_{t-\ell|t-1} = \mathcal{O}_\ell(x_{t-\ell},\bar{u}_{t-\ell|t-1}).
\end{equation}
By the finite-time observability assumption of the system,  $\mathcal{O}_\ell(\cdot,\bar{u})$ admits an inverse for any given $\bar{u}\in \mathcal{U}^{\ell}$. Denote with $\mathcal{O}_\ell^*(\cdot,\bar{u}):\Re^{\ell n_y}\rightarrow \Re^{n_x}$ the inverse map of $\mathcal{O}_\ell(\cdot,\bar{u})$ which is such that 
$\mathcal{O}_\ell^*\left(\mathcal{O}_\ell(x,\bar{u}),\bar{u}\right) = x.$ % \quad \bar{u}\in \Re^{n_x}$$

\noindent It hence follows from \eqref{eq:DataEq} that  
\begin{equation}\label{eq:Inverse-x}
x_{t-\ell}=\mathcal{O}_\ell^*\left(\bar{y}_{t-\ell|t-1},\bar{u}_{t-\ell|t-1}\right)
\end{equation}
which, by recursively applying  the first equation of \eqref{eq_general_form}, gives 
\begin{equation}\label{eq:xt=phi(zt)}
x_t = F_\ell\left(\mathcal{O}_\ell^*\left(\bar{y}_{t-\ell|t-1},\bar{u}_{t-\ell|t-1}\right),\bar{u}_{t-\ell|t-1}\right) \triangleq \phi(z_t). 
\end{equation}
\end{proof}
%%%%
Consider now the more realistic scenario where the (unknown) measurement noise sequence $\left\{w_t\right\}$ is nonzero.
Then Eq. \eqref{eq:DataEq} becomes 
\begin{equation}\label{eq:DataEq-Noise}
\bar{y}_{t-\ell|t-1} = \mathcal{O}_\ell(x_{t-\ell},\bar{u}_{t-\ell|t-1})+\bar{w}_{t-\ell|t-1}.
\end{equation}
As a consequence, the state can no longer be obtained exactly by Eq.  \eqref{eq:Inverse-x} or \eqref{eq:xt=phi(zt)}  since $\bar{y}_{t-\ell|t-1}$ does not lie in the range of $\mathcal{O}_\ell(\cdot,\bar{u}_{t-\ell|t-1})$. 
Let in this case the state ${x}_{t-\ell}$ and $x_t$ be estimated by
\begin{align}
&\hat{x}_{t-\ell} \in \argmin_{x\in \mathcal{X}}\left\|\bar{y}_{t-\ell|t-1}-\mathcal{O}_\ell(x,\bar{u}_{t-\ell|t-1})\right\| \label{eq:xhatPast}\\
&\hat{x}_t =  F_\ell\left(\hat{x}_{t-\ell},\bar{u}_{t-\ell|t-1}\right), \label{eq:xhat}
\end{align}
for some norm $\left\|\cdot \right\|$ on $\Re^{\ell n_y}$. 
The optimization problem  in \eqref{eq:xhatPast} is well-defined since, by Assumptions \ref{assump:Compacts}-\ref{assump:Lipschitz}, the function  $x \mapsto \left\|\bar{y}_{t-\ell|t-1}-\mathcal{O}_\ell(x,\bar{u}_{t-\ell|t-1})\right\|$ is defined on a compact set $\mathcal{X}$ and is continuous.  These, by the extreme value theorem, are sufficient conditions for the existence  of a minimum and for the existence of the minimizer $\hat{x}_{t-\ell}$ as defined above. 

In contrast, the  estimates $\hat{x}_{t-\ell}$ and $\hat{x}_{t}$ need not be uniquely defined in a general setting. Uniqueness would require some more strict conditions on the system under consideration. Here, we will be content with a set-valued version $\hat{\phi}$ of $\phi$ in the noisy estimation scenario.  Hence let 
$\hat{\phi}$ be defined by 
$$
\hat{\phi}(z_t) = \left\{F_\ell\left(\hat{x}_{t-\ell},\bar{u}_{t-\ell|t-1}\right) : \hat{x}_{t-\ell} \mbox{ as in \eqref{eq:xhatPast} }\right\}. 
$$
%\textcolor{red}{Existence of minimizers in \eqref{eq:xhatPast} is not guaranteed unless we request more restrictive conditions. Although the function  $x \mapsto \left\|\bar{y}_{t-\ell|t-1}-\mathcal{O}_\ell(x,\bar{u}_{t-\ell|t-1})\right\|$ clearly admits an infimum, it may not admit a minimum value. To guarantee existence of a minimum, continuity of  $\mathcal{O}_\ell(\cdot,\bar{u}_{t-\ell|t-1})$ would be enough if we restrict attention to compact sets for the state. }
%\medskip
%%%%%%%
A question we ask now is how far the noisy estimate \eqref{eq:xhat} lies from the true state $x_t$. 
To study this, a stronger notion of observability is introduced as follows. 
%%%%%%%
\begin{definition}\label{def:Uniform-Observability}
The system \eqref{eq_general_form} is called \textit{finite-time uniformly observable} over a time horizon $\ell\in \mathbb{N}$ if there exists a constant $\alpha_\ell>0$ such that for each $\bar{u}\in \mathcal{U}^{\ell }$, 
\begin{equation}\label{eq:Uniform-Observability}
\left\|\mathcal{O}_\ell(x,\bar{u})-\mathcal{O}_\ell(x',\bar{u})\right\| \geq \alpha_\ell\left\|x-x'\right\| %\quad \forall (x,x')\in \Re^{n_x}\times \Re^{n_x}
\end{equation}
 for all $(x,x')\in \Re^{n_x}\times \Re^{n_x}$. Here $\left\|\cdot \right\|$ denotes a generic norm defined on appropriate spaces. 
\end{definition}
%%%%
Based on this property, it is possible to bound the error between the noisy estimate \eqref{eq:xhat} and the true state. 
\begin{prop}\label{prop:error-bound}
%Assume that
%\begin{itemize}
%\item The  system \eqref{eq_general_form} is \textit{finite-time uniformly observable} over a time horizon $\ell\in \mathbb{N}$ 
%\item The map $F_\ell$ defined as in \eqref{eq:Fi} is uniformly Lipschitz with respect to the input set $\mathcal{U}$ in the sense that there exists a constant $\gamma_\ell>0$ such that  for all $\bar{u}\in \mathcal{U}^\ell$, 
%$\left\|F_\ell(x,\bar{u})-F_\ell(x',\bar{u})\right\|\leq \gamma_\ell \left\|x-x'\right\|$
%for any pair $(x,x')\in \Re^{n_x}\times \Re^{n_x}$. 
%\end{itemize}
Under Assumptions \ref{assump:Compacts}--\ref{assump:Lipschitz}, if the  system \eqref{eq_general_form} is \textit{finite-time uniformly observable} over a time horizon $\ell\in \mathbb{N}$ in the sense of Definition \ref{def:Uniform-Observability}, then 
\begin{equation}\label{eq:bound}
\left\|\hat{x}_{t}-x_{t} \right\|\leq 2 \gamma_f^\ell \alpha_\ell^{-1}\left\|\bar{w}_{t-\ell|t-1} \right\|
\end{equation}
where $\gamma_f$ is the Lipschitz constant of $f^\circ$ (See Assumption \ref{assump:Lipschitz}) and $\alpha_{\ell}$ is the constant appearing in \eqref{eq:Uniform-Observability}.  
%$$\left\|\hat{x}_{t-\ell}-x_{t-\ell} \right\|\leq 2\alpha_\ell^{-1}\left\|\bar{w}_{t-\ell|t-1} \right\|$$
\end{prop}
%%%
%\vfill
\begin{proof}
It follows from the definition \eqref{eq:xhatPast} of $\hat{x}_{t-\ell}$  that
$$\begin{aligned}
\left\|\bar{y}_{t-\ell|t-1}-\mathcal{O}_\ell(\hat{x}_{t-\ell},\bar{u}_{t-\ell|t-1})\right\|&\\
&\hspace{-2cm}\leq \left\|\bar{y}_{t-\ell|t-1}-\mathcal{O}_\ell(x,\bar{u}_{t-\ell|t-1})\right\|,
\end{aligned}$$
for all $x\in \Re^{n_x}$. In particular this inequality holds for $x=x_{t-\ell}$. By then invoking \eqref{eq:DataEq-Noise} we get 
$$
\begin{aligned}
&\left\|\mathcal{O}_\ell(x_{t-\ell},\bar{u}_{t-\ell|t-1}) -\mathcal{O}_\ell(\hat{x}_{t-\ell},\bar{u}_{t-\ell|t-1})+\bar{w}_{t-\ell|t-1}\right\|\\
&\hspace{4cm}\leq \left\|\bar{w}_{t-\ell|t-1})\right\|.
\end{aligned}
$$
By the triangle inequality property of norms, it follows that 

\begin{gather*}
\left\|\mathcal{O}_\ell(x_{t-\ell},\bar{u}_{t-\ell|t-1}) -\mathcal{O}_\ell(\hat{x}_{t-\ell},\bar{u}_{t-\ell|t-1})+\bar{w}_{t-\ell|t-1}\right\|\\
\geq \left\|\mathcal{O}_\ell(x_{t-\ell},\bar{u}_{t-\ell|t-1}) -\mathcal{O}_\ell(\hat{x}_{t-\ell},\bar{u}_{t-\ell|t-1})\right\| \\
-\left\|\bar{w}_{t-\ell|t-1}\right\|
\end{gather*}

Combining with the previous inequality yields
$$
\begin{aligned}
\alpha_\ell \left\|\hat{x}_{t-\ell}-{x}_{t-\ell}\right\|&\\ 
&\hspace{-2cm}\leq \left\|\mathcal{O}_\ell(x_{t-\ell},\bar{u}_{t-\ell|t-1}) -\mathcal{O}_\ell(\hat{x}_{t-\ell},\bar{u}_{t-\ell|t-1})\right\|\\
&\hspace{-2cm}\leq 2\left\|\bar{w}_{t-\ell|t-1})\right\|.
\end{aligned}
$$
Here, the first inequality is a consequence of the assumption of uniform finite-time observability. 
%On the other hand, the assumption of uniform finite-time observability gives
%$$\left\|\mathcal{O}_\ell(x_{t-\ell},\bar{u}_{t-\ell|t-1}) -\mathcal{O}_\ell(\hat{x}_{t-\ell},\bar{u}_{t-\ell|t-1})\right\|\geq \alpha_\ell \left\|\hat{x}_{t-\ell}-{x}_{t-\ell}\right\|
%$$
As a consequence, $\left\|\hat{x}_{t-\ell}-{x}_{t-\ell}\right\|\leq 2\alpha_\ell^{-1}\left\|\bar{w}_{t-\ell|t-1})\right\| $. 
The result follows now by applying \eqref{eq:xhat}, the uniform Lipschitz assumption on $f^\circ$ stated in \ref{assump:Lipschitz} and Lemma \ref{lem:Lipschitz}.  
\end{proof}
%%%%
It can observed from the expression of the error bound \eqref{eq:bound} that the more strongly the system is observable (that is, the larger the constant $\alpha_\ell$), the more robust the estimate $\hat{x}_t$. 
Indeed $\alpha_\ell$, when it exists, can be defined as 
    $$\inf_{\substack{\bar{u},x,x'\in \mathcal{U}^{\ell}\times \mathcal{X}\times \mathcal{X}\\x\neq x'}}\dfrac{\left\|\mathcal{O}_\ell(x,\bar{u})-\mathcal{O}_\ell(x',\bar{u})\right\|}{\left\|x-x'\right\|}. 
    $$
\section{Modeling and learning}
\label{sec:modelingandlearning}
\subsection{Nonlinear regression model}
\label{subsec:nonlinear_regression_model}

\begin{figure}[t]
    \centering
    \begin{subfigure}[b]{\textwidth}
         \centering
         \includegraphics[width=\textwidth]{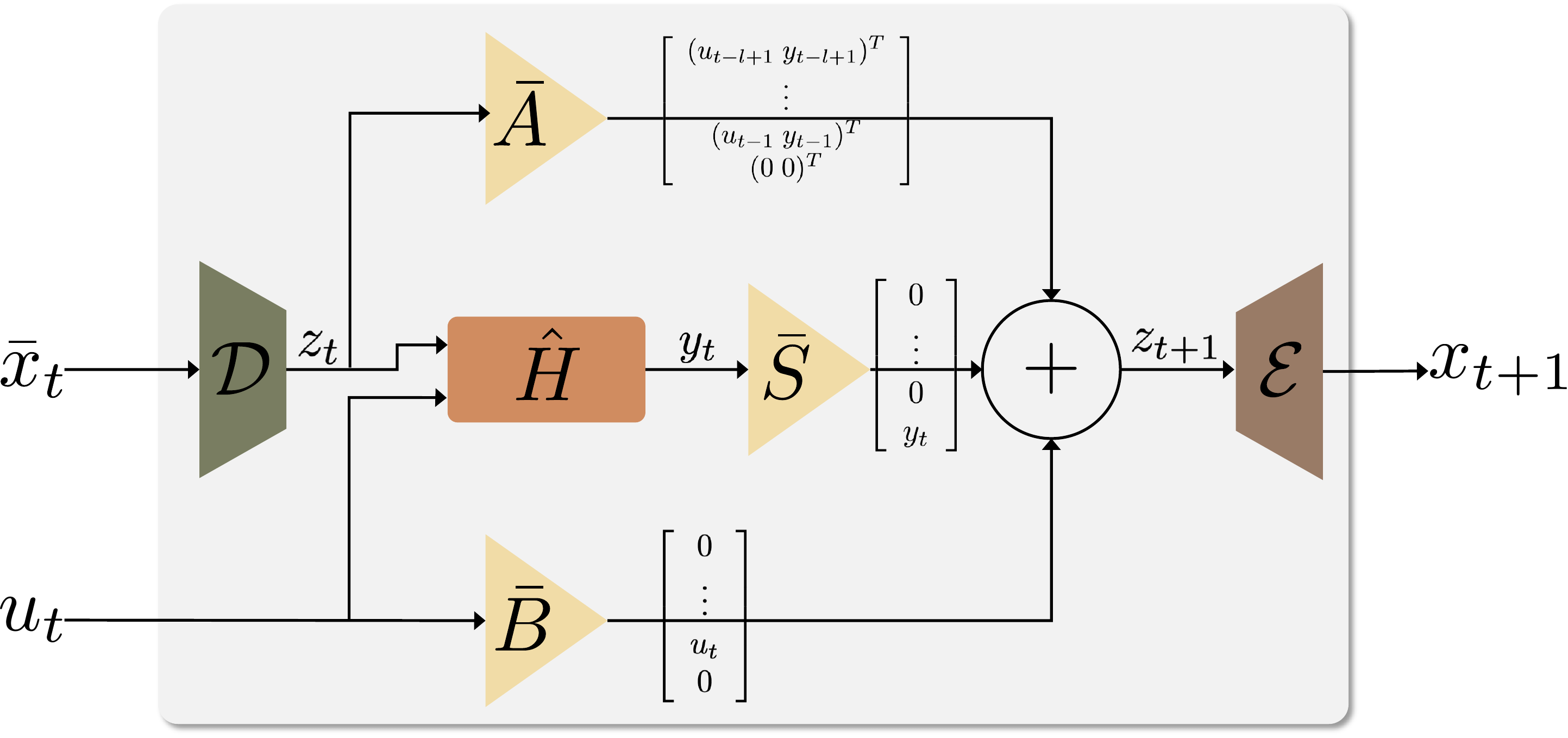}
         \caption{}
         \label{fig:state_space_equation}
     \end{subfigure}
     \begin{subfigure}[b]{\textwidth}
         \centering
         \includegraphics[width=\textwidth]{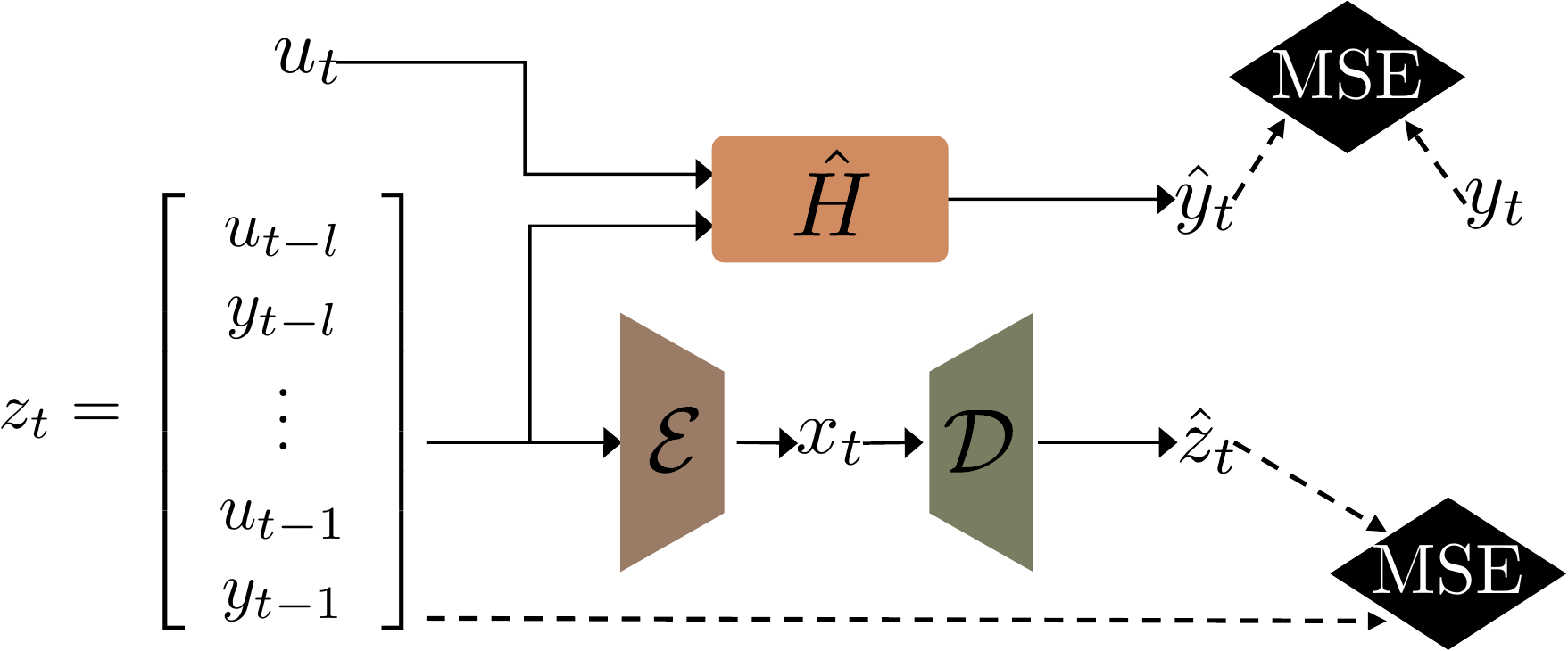}
         \caption{}
         \label{fig:learning_models}
     \end{subfigure}
    \caption{(a) Block diagram of our canonical reduced state space representation as defined in \eqref{eq:Final-State-Space-Estimate} (b) Training: the dynamic is modeled by $\hat H$ acting as a regressor from a short history of previous observations to future. The encoder-decoder model is used to reduce the size of the state $z_t$ to $x_t$. We train each network by minimizing the prediction error as well as the reconstruction error.    \vspace{-5mm}}
    \label{fig:model}
\end{figure} 
\noindent   
A starting point of our identification method for system  \eqref{eq_general_form} is to solve a nonlinear regression problem. 
To formulate this problem, note by Proposition \ref{prop:error-bound} that the true state of the system can be written as $x_t=\hat{x}_t+\delta_t$ with $\left\|\delta_t\right\|\leq \gamma_f^\ell \alpha_\ell^{-1}\left\|\bar{w}_{t-\ell|t-1} \right\|$.  Consider now plugging the state estimate \eqref{eq:xhat} into the output equation of \eqref{eq_general_form}, which gives
$$
\begin{aligned}
y_t&= h^\circ(\hat{x}_t+\delta_t,u_t)+w_t\\
& = h^\circ(\hat{x}_t,u_t)+\xi_t,
\end{aligned}
$$
with $\xi_t$ being an error component entirely due to the noise. It is indeed equal zero whenever $w\equiv 0$. 
It can be shown that $\xi_t$ can be written as $\xi_t=w_t+\tilde{\delta}_t$ with $\|\tilde{\delta}_t\|\leq \gamma_h\gamma_f^{\ell} \alpha_\ell^{-1}\left\|\bar{w}_{t-\ell|t-1} \right\|$, where $\gamma_h$ is the Lipschitz constant of the measurement  function $h^\circ$ of system  \eqref{eq_general_form}. 
Since $\hat{x}_t$ is a  function of $z_t$ we end up with 
\begin{equation}\label{eq:regression}
y_t = H^\circ(z_t,u_t)+\xi_t
\end{equation}
for some nonlinear function $H^\circ$. 
\begin{remark}
In the absence of noise, the exact expression of $ H^\circ$ is  
$ H^\circ(z_t,u_t) =h^\circ\Big( F_\ell\big(\mathcal{O}_\ell^*(z_t),\eta(z_t)\big),u_t\Big)$
with $\eta(z_t)=\bar{u}_{t-\ell|t-1}$. %\blkdiag(0_{\ell n_y},I_{\ell n_u})z_t=
\end{remark}

The first step of the identification  method is to construct a high dimensional state-space representation whose state is the vector  $z_t$ defined in \eqref{eq:zt}. 
More precisely, consider
\begin{equation}\label{eq:state-space-Z}
\left\{\begin{aligned}
&z_{t+1}= \bar{A}z_t+\bar{B} u_t+\bar{S}   H^\circ(z_t,u_t)+\bar{S} \xi_t\\
& y_t = H^\circ(z_t,u_t)+\xi_t,
\end{aligned}
\right.
\end{equation}
where $\bar{A} = A\otimes I_{n_u+n_y}$, $\bar{B} = e_{n_x-1}\otimes I_{n_u}$, $\bar{S} = e_{n_x}\otimes I_{n_y}$, $e_i\in \Re^{n_x}$ being the canonical basis vector which has $1$ in its $i$-th entry and zero everywhere else, $\otimes$ referring to the Kronecker product  and $A\in \Re^{n_x\times n_x}$ given by the canonical form 
$$
A = \begin{pmatrix}
0 & 1 &  0& \cdots & 0\\
 \vdots &  \ddots &\ddots &\ddots& \vdots \\
0  &  \cdots &\ddots &1 & 0\\
 0& \cdots &  \cdots & 0&  1\\
 0&\cdots  &\cdots  & 0 &0
\end{pmatrix}.
$$
From \eqref{eq:regression} it can be seen that \eqref{eq:state-space-Z} constitutes a state-space representation for system  \eqref{eq_general_form} since both models have the same input-output behavior for $t\geq \ell$. 

Given a finite set of input-output data points $\left\{(u_t,y_t)\right\}_{t=1}^{T+\ell}$, an estimate of the function $ H^\circ$ can be obtained in a certain nonlinear model class $\mathcal{H}$ as 
$$\hat{H}\in \argmin_{H\in \mathcal{H}}{J(H)}, $$
where $J(H)$ is a regression loss given as
\begin{align}
    J(H) = \frac{1}{T} \sum_{t=\ell+1}^{T+\ell} \alpha_t\| y_{t} - H(\hat{z}_{t}, u_{t})\|^2 \\
    \text{ s.t. } \hat{z}_{t+1} = \bar{A}\hat{z}_t + \bar{B}u_t + \bar{S}H(\hat{z}_t, u_t),\; \hat{z}_{\ell+1}=z_{\ell+1},
\end{align}

\noindent
where $\alpha_t$ is a weighting coefficient such as $\alpha_t=1$ except for $\alpha_{\ell+1}=10$. Regression starts after a burn-in phase of $\ell$ steps (ie. the window size), which are needed to construct a full state-representation.

The choice of the model class $\mathcal{H}$ is fundamental for several reasons: it must be sufficiently large to represent a good approximation of $H^\circ$, and it should not be too large in order to ensure learnability and generalization to unseen conditions \cite{BookShalev2004}. In other words, this class needs to be complex enough to capture the behavior of the unknown nonlinear function $H^\circ$ while being easily identifiable from rather limited set of measurements. We rely on the well-documented approximation power of deep neural networks and propose to select $\mathcal{H}$ to be the class of multi-layer perceptrons (MLP), whose capacity shall be limited. In particular, the number of hidden layers and number of neurons of each layer is considered as a hyper-parameter optimized over a validation set independent of training and evaluation splits, as classically done in machine learning.
%with 3 hidden layers of 256 units each. 

\subsection{Model reduction}
\noindent
The last step of the proposed method consists in model reduction. In effect, the model described in \eqref{eq:state-space-Z}, although structurally simple, may suffer from a high dimensional state vector $z_t$. This  may be a concern for some applications.  We therefore propose a second deep learning structure allowing nonlinear state-space model reduction by encoding the state variable $z_t$ into a low dimensional state variable $\bar{x}_t\in \Re^{\bar{n}}$ for some user-defined dimension $\bar{n}\in \mathbb{N}$. Formally, we train an auto-encoder $(\mathcal{E}, \mathcal{D})$ such that $\bar{x}_t = \mathcal{E}(z_t)$ and $z_t = \mathcal{D}(\bar{x}_t)$. By applying these maps to Eq. \eqref{eq:state-space-Z} and neglecting the noise terms, we get an approximate representation of the initial system \eqref{eq_general_form} as follows: 
\begin{equation}\label{eq:Final-State-Space-Estimate}
    \left\{\begin{array}{l}
        \bar{x}_{t+1}  = \mathcal{E}\Big(\bar A\mathcal{D}(\bar{x}_t) + \bar Bu_t + \bar S \hat{H}(\mathcal{D}(\bar{x}_t), u_t)\Big) \\
        \bar{y}_t = \hat{H}(\mathcal{D}(\bar{x}_t), u_t).
    \end{array}\right.
\end{equation}
The state space equation is summarized in figure \ref{fig:state_space_equation}. The parameters of the encoder $\mathcal{E}$ and decoder $\mathcal{D}$ are trained with the following reconstruction loss from data samples $\{z_t\}$ collected from the training set (see also Fig. \ref{fig:learning_models}).
\begin{equation}
\mathcal{(E,D)} = \arg \min_{\mathcal{E',D'}} 
\left \|
z_t - \mathcal{D'}(\mathcal{E'}(z_t))
\right  \|
\end{equation}
Note that  the dimension $\bar{n}$ of the compressed state $\bar{x}_t$ in the estimated model \eqref{eq:Final-State-Space-Estimate} is potentially different from the true state dimension  $n_x$.  
Another observation is that by going from  \eqref{eq:state-space-Z} to \eqref{eq:Final-State-Space-Estimate}, one reduces the dimension of the state vector but at the cost of introducing some structural complexity. Hence the computational price associated with simulating a model such as \eqref{eq:Final-State-Space-Estimate} may still be high depending on the complexity of the auto-encoder $(\mathcal{E}, \mathcal{D})$.  

From a formal point of view, this reduced state has several advantages. The constructed state-space $z_t$ is not part of the update equation (\ref{eq:Final-State-Space-Estimate}) anymore. We can also experimentally show (see section \ref{sec:experiments}), that this method can discover state representations of smaller size %than the complete system 
with a method which is generic in nature and can be applied to a broad class of problems.

%the dimension $n_x$ of the state $x_t$ in \eqref{eq_general_form} and

% \smallskip
% ...........................................................

% \noindent More precise description of the model class to be given here

% .............................................................

\medskip
\noindent

\section{Experimental results}
\label{sec:experiments}
\begin{figure*}[t]
    \centering
    \includegraphics[width=\textwidth]{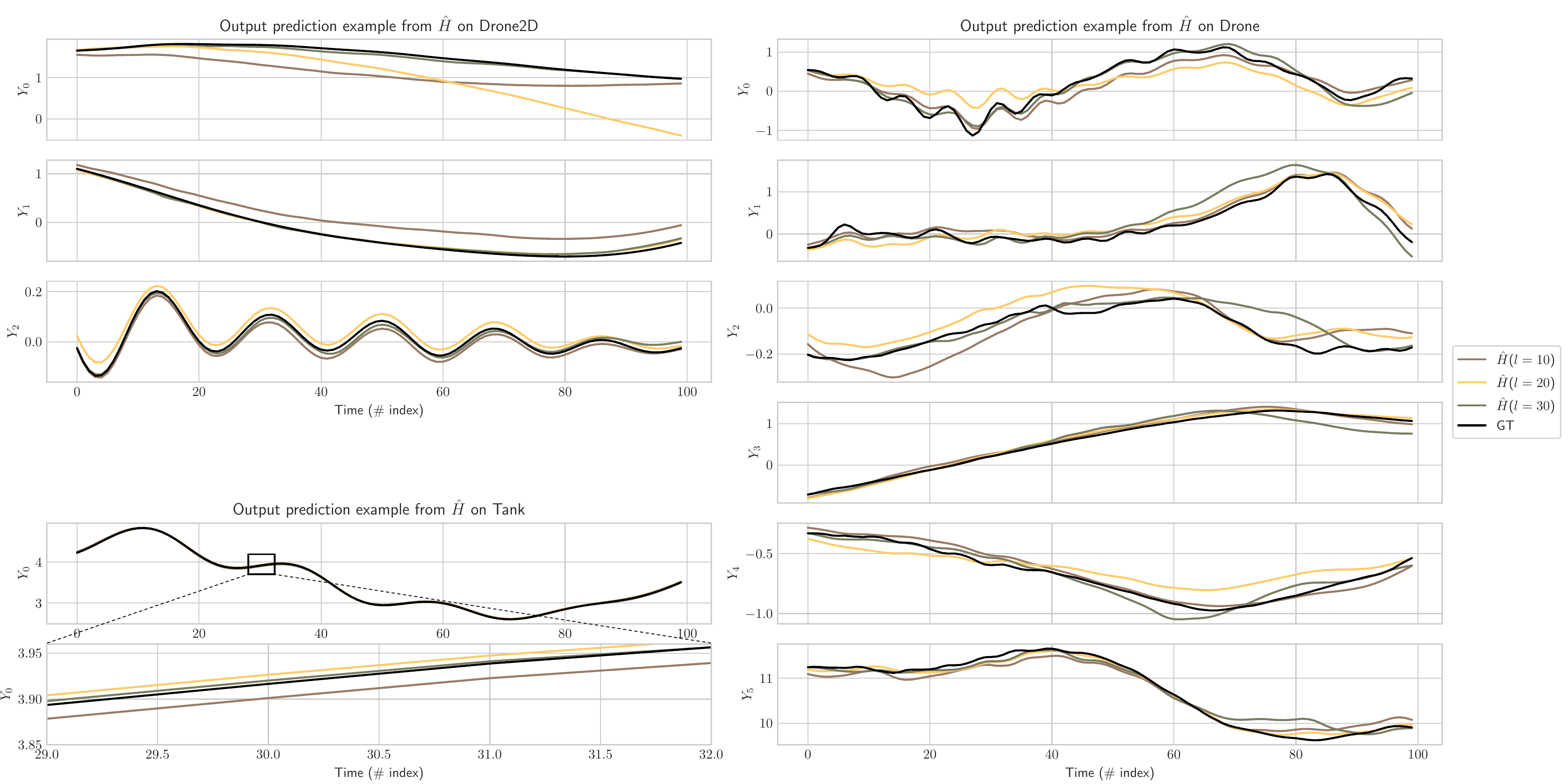}
    \caption{Visual example of the output prediction made by $\hat{H}$ (\textbf{Ours (MLP)}) for different value of $\ell$ on the three datasets. 
    \vspace{-5mm}
    }
    \label{fig:rollout}
\end{figure*}

\noindent
We illustrate and evaluate the proposed nonlinear dynamical model identification approach on the estimation and prediction of the state of a system with unknown dynamics.
%The proposed nonlinear dynamic model identification approach will be illustrated for estimating and predicting the state of a system whose dynamics is unknown. 
To demonstrate the practical feasibility of our model, we propose to study its behavior in three different scenarios. First, we demonstrate the capabilities of our regressor function $H^\circ$ for output prediction on simulated systems.\footnote{Code and datasets will be made publicly available upon acceptance.} We also study the influence of key parameters, namely the length $\ell$ of the time window and the impact of the state reduction.
%The proposed nonlinear dynamic model identification approach will be illustrated on too examples : a simple benchmark problem and real data ...

\subsection{Dynamical systems and benchmarks}
%\johan{Papier intéressant : \cite{beintema2021}. Se compare à Masti et Bemporad, utilise 2 jeux de données publiques.}
\noindent
We use two simulated and one real system to validate our contributions. 
\begin{description}[leftmargin=0.4cm]
\item[Tank] --- we test the proposed method on the cascade tank system introduced in \cite{schoukens2017three}. This system relates the water level in two connected tanks without consideration of overflow. It has the form  \eqref{eq_general_form} of a discrete-time state-space model with $f^\circ$ and $h^\circ$ implicitly instantiated as follows 
\begin{equation}
    \left\{\begin{array}{l}
         x_{1,{t+1}} =  x_{1,t} - k_1\sqrt{x_{1,t} } + k_2u_t  \\
         x_{2,t+1} = x_{2,t} +k_3\sqrt{x_{1,t} }- k_4\sqrt{x_{2,t} }  \\
         y_t = x_{2,t}+w_t,
    \end{array}\right.
\end{equation}
with   $x_t=\begin{pmatrix} x_{1,t} &  x_{2,t}\end{pmatrix}\in \Re^2$ being the state and $k_i$, $i=1,\ldots,4$ known parameters.  
\item[2D Drone] --- we introduce a model of a 2-dimensional drone, i.e.  an unmanned aerial vehicle, which moves in a 2D plane. The drone is equipped with two propellers and its dynamic is modeled by:
\begin{equation}
    \left\{\begin{array}{l}
        \ddot p_x = -\frac{k_T}{m}(\Omega_1^2+\Omega_2^2)\sin(\theta) - \frac{\gamma}{m}(\Omega_1+\Omega_2)\dot p_x \\
        \ddot p_z  =  \frac{k_T}{m}(\Omega_1^2 + \Omega_2^2)\cos(\theta) - \frac{\gamma}{m}(\Omega_1+\Omega_2)\dot p_z - g \\
        \ddot \theta = \frac{k_TL}{J}(\Omega_2^2 - \Omega_1^2) \\
        y  = (p_x\quad p_z \quad \theta)^T,
    \end{array}\right.
\end{equation}
% \lb{possible conflit de notation entre les coordonnées spatiales $x$ et $z$ et le vecteur d'état $x$ et le regresseur $z$. Peut-on changer en quelque chose comme $s_x$ et $s_z$ ou autre chose? }

% \mn{et d'ailleurs, s z  $\theta  \gamma$ ne me semble pas être définies. on peut noter $p_x$ et $p_z$ pour position en x et en z?}

where $(p_x, p_z)$ is the position, $k_T$ the thrust constant, $\Omega_i$ the rotationnal speed of the $i^{th}$ propeller, $L$ the length of the UAV, $m$ its mass, $J$ its inertia and $\gamma$ its friction coefficient. The main interest of such a system is its naturally unstable dynamics, which complicates the 
identification process. The system has been discretized.
\item[3D Drone] --- 
We also evaluate on recordings of the Blackbird UAV flight dataset \cite{antonini_blackbird_2018}, which consists of 10 hours of aggressive quadrotor flights, measured with an accurate motion capture device. We use this real world data to demonstrate that our observer discovers a state representation containing the same information as in the physical state without any supervision.

We processed the raw data gathered from the on-board inertial measurement unit (IMU) and propeller rotation speeds as observation and command signals. The regressor is trained to simulate the IMU measurements, i.e. acceleration and angular speed of the drone expressed in the local frame. 
\end{description}

\noindent
Noise has been added to the observations for the two simulated systems, tank and 2D Drone. More details about the dataset generation is provided in the appendix.

\subsection{Baseline methods}
\noindent
To experimentally compare our model to competing approaches from the literature, we introduce a neural baseline in the form of gated recurrent units (GRU) \cite{GRU2014}, the state-of-the-art variant of recurrent neural networks. This is a pure data-driven technique from the machine learning field, where the learned state representation directly corresponds to the hidden state vector of the GRU. For a fair comparison, we limit the size of the hidden vector to fit the corresponding size of $z_t$. We refer to this model as \textit{classic GRU}. Its update equations are given by
\begin{equation}
    \begin{array}{ll}
         h_{t+1} &= \text{GRU}([y_t;u_t], h_t), \quad h_{0} = 0 \\
         y_t &= \text{MLP}(h_t),
    \end{array}
\end{equation}
where $\text{GRU(.)}$ is a shorthand notation corresponding to the classical update equations of GRUs \cite{GRU2014}. For simplicity, and as usually done, gates have been omitted from the notation.

The baseline is evaluated in a setting which is comparable to the proposed model. In particular, the model has access to the same window of input/output pairs $[y_{t+k}, u_{t+k}]_{k=1..\ell}$ during the initial burn-in phase. However, these values do not explicitly make up the state, as in our model. This data-driven baseline is sufficiently general to be able to learn the same state representation in theory, but there is no guarantee that training will lead to this solution.

We also experiment with the model introduced in \cite{MASTI2021} which consists of an auto-encoder with a learned latent dynamics that operates on the reduced state representation. This model has been evaluated on the same tank system, yet, with a different data collection technique. Train and test trajectory in the Tank dataset as proposed in \cite{MASTI2021} are generated from PRBS-like signals, which is a classical approach for system identification. Our version of Tank dataset is much more challenging: observations are collected from closed-loop simulations with targets generated in a procedural manner and PID control. In our dataset, we took care to explore a wide range of possible states with sparse measurements in the train set to prevent over-fitting on a specific command design.

% \lb{On pourrait envisager de supprimer la phrase surlignée ci-dessus en rouge. Au cas où Masti ou Bemporad reviewent notre article, s'ils ne sont pas d'accord avec nous, on pourrait avoir des problèmes. }
% S : Oui tu as raison

\subsection{Extension: a hybrid state-space model}
\noindent
We also introduce an extension of our model, which combines the advantages of both methodologies. It uses our proposed state representation $z_t$, but implements the mapping $H^\circ$ by a GRU in place of the MLP proposed in section \ref{sec:modelingandlearning}.
%, which means that $H^\circ 
Formally, the GRU updates a zero-initialized hidden vector using the previous observations and command. This vector is then decoded by a MLP to the desired observation. Equation \eqref{eq:state-space-Z} is then used for forward prediction. We refer to this model as \textit{Ours (GRU)}, it is given as follows:

\begin{equation}
    \begin{array}{ll}
         z_{t+1} &= \bar A z_t + \bar B u_t + \bar S \ \text{MLP}(h_t) \\
         h_{i+1} &= \text{GRU}([y_i;u_i], h_i), \quad h_{t-l} = 0 \\
         y_t &= \text{MLP}(h_t).
    \end{array}
\end{equation}

\subsection{Output prediction and parameter analysis}

\begin{table*}[t]
\footnotesize
\begin{tabular}{@{}c|cccc|cccc|cccc@{}}
\toprule
\multirow{2}{*}{\shortstack{Window\\ size}} 
        & \multicolumn{4}{c|}{Tank ($\times 10^{-4}$)}    & \multicolumn{4}{c|}{2D Drone ($\times 10^{-2}$)} & \multicolumn{4}{c}{3D Drone ($\times 10^{-2}$)} \\ \cmidrule(l){2-13} 
        $\ell$ & \shortstack{Classic \\GRU$^\dagger$} & \shortstack{Masti \\et al.\cite{MASTI2021}} & \textbf{\shortstack{Ours \\(GRU)}}   &  \textbf{\shortstack{Ours \\(MLP)}}  &  \shortstack{Classic \\GRU$^\dagger$} &  \shortstack{Masti \\et al.\cite{MASTI2021}}     &\textbf{\shortstack{Ours \\(GRU)}} &  \textbf{\shortstack{Ours \\(MLP)}} &  \shortstack{Classic \\GRU$^\dagger$} &   \shortstack{Masti \\et al.\cite{MASTI2021}}     & \textbf{\shortstack{Ours \\(GRU)}} & \textbf{\shortstack{Ours \\(MLP)}}\\ \midrule
5       & 163   & 1030 & 138             & \textbf{7.14}  & 62.8 & 60.5 & 106              & \textbf{31.4}   & 24.8 & 14.6 & \textbf{6.44} & 15.2  \\
10      & 41.7  & 1070 & 5.60            & \textbf{0.930} & 82.7 & 58.6 & 68.9             & \textbf{9.95}   & 23.7 & 14.5 & \textbf{5.32} & 14.2  \\
15      & 4.57  & 957  & 3.06            & \textbf{0.960} & 61.9 & 58.2 & 35.2             & \textbf{7.52}   & 23.4 & 13.0 & \textbf{5.07} & 13.6  \\
20      & 4.04  & 914  & 1.07            & \textbf{0.761} & 78.4 & 55.3 & 19.3             & \textbf{8.06}   & 22.7 & 12.6 & \textbf{4.68} & 13.5  \\
25      & 0.600 & 915  & \textbf{0.481}  & 0.606          & 80.3 & 53.6 & 23.0             & \textbf{5.17}   & 21.3 & 12.1 & \textbf{4.83} & 13.6  \\
30      & 1.73  & 917  & \textbf{0.193}  & 0.448          & 104  & 51.3 & 25.0             & \textbf{3.13}   & 19.2 & 12.6 & \textbf{4.61} & 13.1 \\
\midrule
\multicolumn{13}{l}{$^\dagger$ {\footnotesize The size of the hidden state of each GRU model is adapted to the window size s.t. fits the size of the equivalent regressor model.}}\\ 
\bottomrule
\end{tabular}
\caption{\label{tab:regression_rollout}Quantitative evaluation: we report MSE error over 100-step rollouts by the learned regression model and compare with baselines, for different windows sizes $\ell$. Our model consistently outperforms all baselines.\vspace{-5mm}}
\end{table*}

\textbf{Output forecasting} -- the identified dynamic model can be evaluated by performing open-loop forward prediction from initial conditions and the set of inputs applied to the real system. The model then forecasts outputs, which may be compared to actual ground truth measurements. We assessed the first stage of our method using this task, i.e. the resolution of the regression problem. Table \ref{tab:regression_rollout} reports the mean squared error on 100 step roll-out predictions for each baseline and different window sizes $\ell \in \{5, 10, 15, 20, 25, 30\}$. Our method shows excellent prediction error even for low window sizes, and consistently outperforms the closest competing method from the literature, Masti et al. \cite{MASTI2021}, by a large margin. We conjecture two key arguments to justify this difference : (1) the structure proposed by \cite{MASTI2021} suffers from complex interaction between the auto-encoder and the latent dynamics that penalizes learning, and (2) the model design process over-fitted on the simpler dataset used in the original paper.

\textbf{Machine learning baseline} -- is competitive with our contribution. However, its structure forces to observe only one couple $(y_t, u_t)$ at a time. Relevant information needs to be stored in its memory, the vectorial hidden state, and this storage process is fully learned by gradient descent, a difficult process. In principle, these models can learn a state representation which is similar or even identical to our designed state-map, but there is no guarantee that this representation emerges. Our state map model can therefore be seen as a form of useful inductive bias for recurrent neural models. 

For moderate window sizes, our model benefits from the immediate availability of all the components of $z_t$ in its state. For very large window sizes or complex dynamical systems (such as 3D Drone), the GRU extension (\textit{Ours (GRU))} outperforms the MLP regressor. In this situation, the GRU takes advantage of its incrementally updated memory, and manages to manipulate the large dimension of $z_t$ by processing it piecewise, whereas the MLP must manipulate the entire vector. Figure \ref{fig:rollout} shows samples of predicted trajectory using the MLP regressor approach for each dataset.

%, when this occurs the prediction error is of a very small order of magnitude, making the difference between the models negligible.

% \begin{figure*}[t]
%   \includegraphics[width=0.95\textwidth]{figures/H0_rollout_barplot.pdf}
%   \caption{\label{fig:parameter_influence} We measure the performance of $\hat{H}$ (estimate of $H^\circ$) on the synthetic systems performing open-loop roll-out prediction on 100 steps in the future for different window sizes, i.e. the number of previous observations embedded in $z_t$. We report the mean squared error between the predicted observations $\hat y_{0:100}$ to the ground truth measurements $y_{0:100}$.}
% \end{figure*}

% \textbf{Window size} -- is, as expected, a sensible parameter of our method, which is positively correlated with the prediction accuracy. \sj{Je sais pas quoi dire sur cette figure, mais je veux pas la virer parce que c'est le seul endroit où on présente des résultats sur drone 3D.}.

\subsection{Model reduction}
\noindent 
The reduction step is performed downstream of the regression model training. Nevertheless, the difficulty of the reduction task is directly related to the initial size of the state representation $z_t$, that is, to the size of the window $\ell$. In order to accurately evaluate the compression capabilities of our approach, we trained several auto-encoders for each value of $\ell \in \{5, 10, 15, 20, 25, 30\}$ corresponding to different rates of compression increasing by steps of 15\%.

Figure \ref{fig:autoencoder_heatmap} shows the compression capabilities of our encoder-decoder structure for different  window sizes $\ell$. The results are consistent on the three datasets. The compression rate is more sensitive on small input dimensions, and conversely, a larger dimension can be reduced extensively with negligible loss of accuracy. Indeed, increasing the number of inputs arguably leads to an increase in the redundancies exploitable by the encoder to reduce the dimension of the state space and reconstruct it with limited deviation with respect to the initial vector.

Yet such reduction introduces noise to the state representation that the regressor will have to cope with. We thus evaluate the impact of state-space reduction on the output forecasting capabilities of our model, and summarize the results in figure \ref{fig:state_compression}. Our reduction method manages to reduce the dimension of the state in a consistent way up to 60\% for the two datasets in simulation without sensible variation of the prediction error. The error bars reflect the double dependence of our approach both on the performance of the regression model $H^\circ$ but also on the quality of the encoding-decoding. We compare favorably to the baseline in \cite{MASTI2021}.

\begin{figure}[t]
    \centering
    \includegraphics[width=\textwidth]{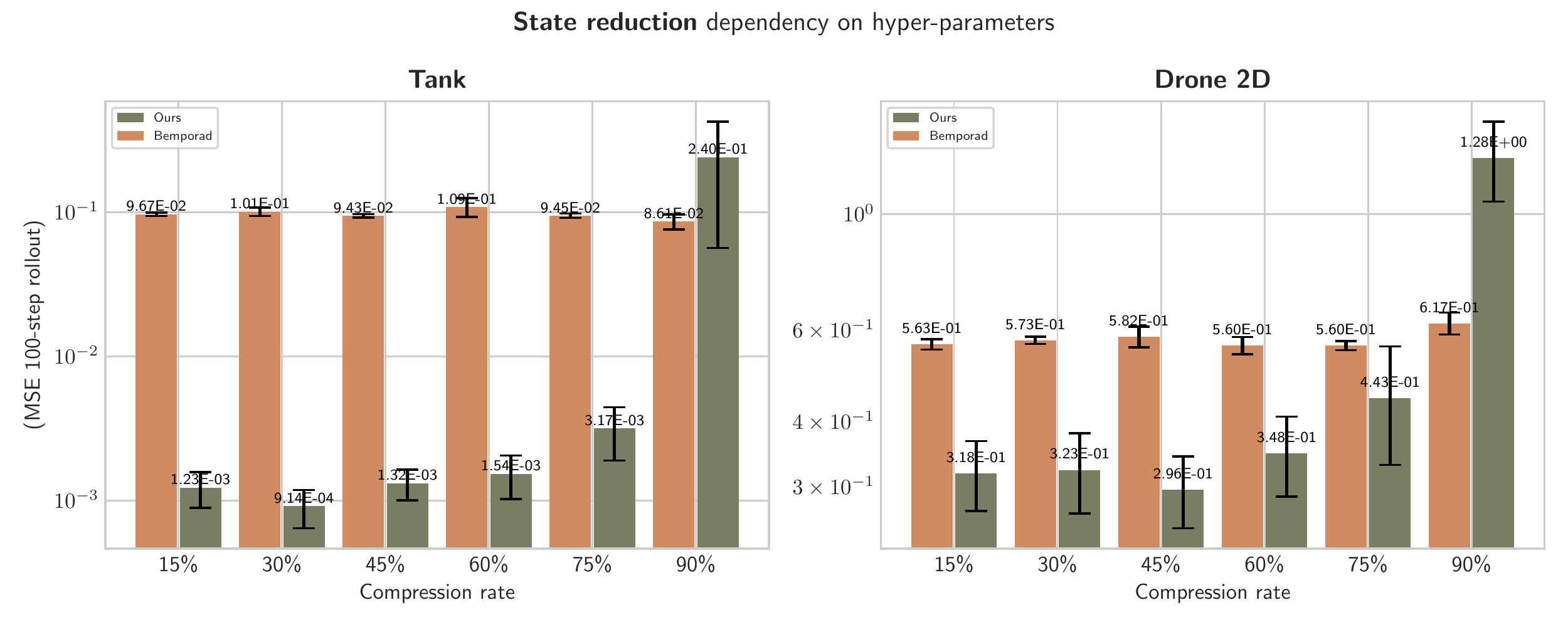}
    \caption{We studied the impact of state compression for multiple configurations of final latent state dimension and temporal window and aggregate the results by this two paramaters on the synthetic datasets. Specifically, we measure the MSE on observation prediction error for 100 step in the future (\textit{Ours (MLP)}).    \vspace{-3mm}}
    \label{fig:state_compression}
\end{figure}

\begin{figure}
    \centering
    \includegraphics[width=\textwidth]{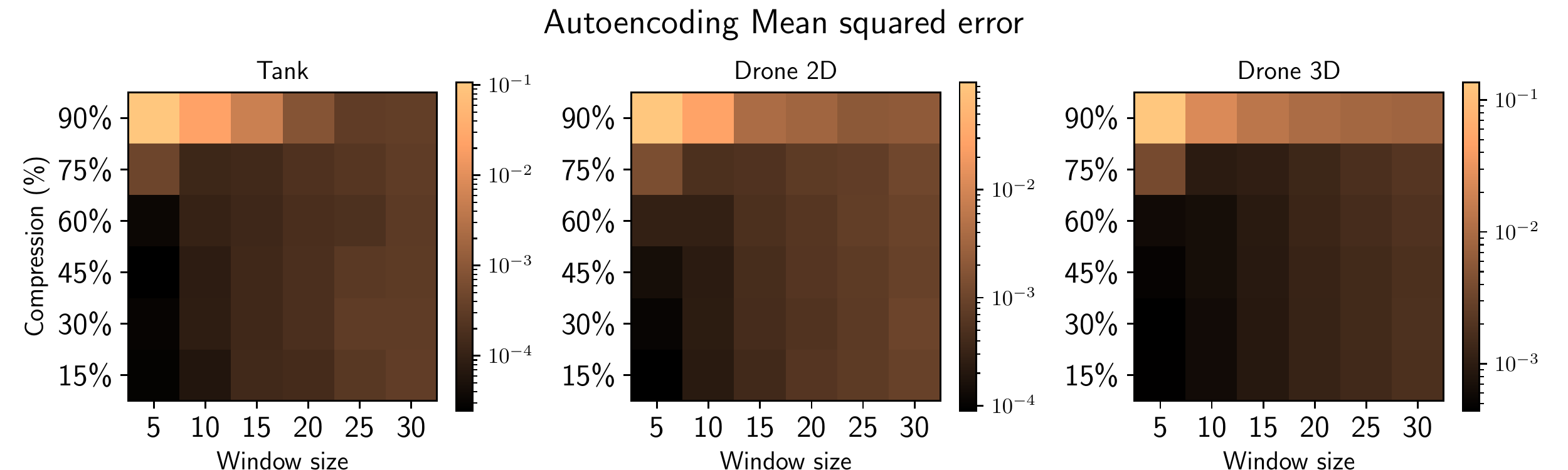}
    \caption{Heatmap of MSE for the encoder-decoder structure depending on both the window size (which relates to the initial state dimension $\text{dim }z$) and the compression rate.\vspace{-5mm}}
    \label{fig:autoencoder_heatmap}
\end{figure}

% \subsection{Drone}

% The use and the research field of Unmanned Aerial Vehicle or UAV has grown considerably for the past decades. Nowadays their versatility made them very useful in many sectors such as  the cinema, exterior maintenance, rescue mission or the military, especially quadrotors. Propelled with four rotors their symmetry centers theirs mass and inertia. However their control is heavily impacted with air turbulences.

% Usually when facing perturbations control is used. 
% This work \cite{berkenkamp_safe_2016} focus on designing a controller with a Gaussian process learning the perturbations.
% Different works aim to model this turbulences with different techniques. While \cite{cheeseman_effect_1955} \cite{perozzi_trajectory_2018} models them, \cite{bauersfeld_neurobem_2021} \cite{shi_neural_2019} learns them with past drone states and commands using neural networks.

% To work with real turbulences this work is based on real drone data from two datasets:  \cite{burri_euroc_2016} who use a quadrotor and  \cite{antonini_blackbird_2018} who use an hexarotor.

\section{Conclusion}
\noindent
In this work, we take advantage of the power of high-capacity deep neural networks to design a new methodology for estimating nonlinear dynamical systems from a set of input/output data pairs. We show that the state can be expressed as a state map computed as a function of past inputs and outputs. We learn a mapping from this representation to model outputs from training data using deep networks and show that this approach is competitive.
%We first give conditions to theoretically ensure the correct definition of certain optimization problems used in the learning algorithm, which hold under moderate assumptions, we Then ...
We tackled the problem of reducing the state space, showing that this way a state of similar size than the original problem can be obtained through machine learning with an auto-encoding solution. The proposed approach can be used to reduce the order of a given nonlinear model, such as  infinite-dimensional discrete systems. The methodology was validated using three numerical examples and using a data-set from real experiments from the literature.

\bibliography{references}

\appendix
\subsection{Dataset details}
\noindent
\textbf{Tank} -- dataset is generated by uniform sampling of five waypoints lying in $[0, 5]$ evenly distributed on time on a $200$ steps reference constructed by cubic spline interpolation between the waypoints. This reference is then tracked with a PID controller. Our dataset contains $60$ trajectories for the train set, and $20$ for both the validation and test set. For simulation, we use $k_1=0.5, k_2=0.4, k_3=0.2$ and $k_1=0.3$. 

\noindent
\textbf{3D drone} -- dataset built on the BlackBird dataset \cite{antonini_blackbird_2018}. We extracted IMU measurements and commands from raw data and apply pre-processing as follows : temporal synchronization of both signals, noise filtering have using Butterworth filters, and sampling rate reduction to $50$Hz. To create train/valid/test splits, we sampled 20 flights to form the validation and 10 for test split. The remaining 146 flights were used for the training set. Each flight have been sliced in 200-steps chunks to facilitate training.

\noindent
\textbf{2D drone} -- dataset is generated by uniform sampling of 5 to 10  2D waypoints lying in $[-2, 2]$ evenly distributed on time on a 600 steps reference constructed by cubic spline interpolation between the waypoints. This reference is then tracked with an model predictive control approach. Our dataset contains $500$ flights for training, and $20$ flights for validation and test sets. For simulation, we choose $k_T= 4\times 10^{-4}, \gamma=10^{-9}, L=0.15, m=1$ and $J=2.7\times 10^{-3}$. The system is simulated with Euler integration scheme at $30$Hz.

\subsection{Models details}
\noindent
\textbf{Classic GRU} -- is a 2 layer gated reccurent unit. The hidden vector size is chosen such that the cumulated dimension of the two hidden vectors matches the one of the corresponding state $z_t$, formally $n_h = \frac{l}{2} (n_u + n_y)$. The hidden vector is then decoded by a multilayer perceptron with one hidden unit of size $n_h$. 

\noindent
\textbf{Ours (MLP)} -- uses a MLP to model $H^\circ$ with 3 hidden units of size 256 for simulated dataset and 2 layer with 2048 units for the 3D Drone. The encoder-decoder is modeled with 2 MLP with 2 layers of 512 units.

\noindent
\textbf{Ours (GRU)} -- model $H^\circ$ with a GRU with three layers, and hidden size of 128. The encoder-decoder is identical as Ours (MLP).

Each model is implemented in Pytorch and trained with Adam optimizer, with learning rate of $10^{-4}$. We trained the regressor (both MLP and GRU) for 10,000 epochs, and the encoder for 3,000 for the simulated datasets and respectively 300 epochs for the 3D drone dataset.

\end{document}